\documentclass[a4paper,onecolumn,11pt,accepted=2024-10-11]{quantumarticle}
\pdfoutput=1

\usepackage[english]{babel}
\usepackage[T1]{fontenc}
\usepackage{graphicx}
\usepackage{dcolumn}
\usepackage{bm}
\usepackage{hyperref}
\usepackage{placeins}
\usepackage{xcolor}
\usepackage{dsfont}
\usepackage{amsmath}
\usepackage{amssymb}
\usepackage{amsthm}
\usepackage{soul}
\newtheorem{theorem}{Theorem}
\newtheorem{lemma}[theorem]{Lemma}
\newtheorem{definition}[theorem]{Definition}

\begin{document}

\title{Limitations for Quantum Algorithms to Solve Turbulent and Chaotic Systems}

\author{Dylan Lewis}
\affiliation{Department of Physics and Astronomy, University College London, London WC1E 6BT, United Kingdom}

\author{Stephan Eidenbenz}
\affiliation{Los Alamos National Laboratory, Los Alamos, NM, USA}

\author{Balasubramanya Nadiga}
\affiliation{Los Alamos National Laboratory, Los Alamos, NM, USA}

\author{Yi\u{g}it Suba\c{s}\i}
\affiliation{Los Alamos National Laboratory, Los Alamos, NM, USA}

\begin{abstract}
We investigate the limitations of quantum computers for solving nonlinear dynamical systems. In particular, we tighten the worst-case bounds of the quantum Carleman linearisation (QCL) algorithm [Liu et al., PNAS 118, 2021], answering one of their open questions. We provide a further significant limitation for any quantum algorithm that aims to output a quantum state that approximates the normalized solution vector. Given a natural choice of coordinates for a dynamical system with one or more positive Lyapunov exponents and solutions that grow sub-exponentially, we prove that any such algorithm has complexity scaling at least exponentially in the integration time. As such, an efficient quantum algorithm for simulating chaotic systems or regimes is likely not possible.
\end{abstract}

\maketitle

\section{Introduction}
Finding, modelling, and solving differential equations constitutes a significant and essential part of science, with foundational physics such as quantum mechanics being described through the Schrödinger and Dirac equations, electromagnetism through Maxwell's equations, the classical physics of sound and heat through the wave and diffusion equations, the dynamics of fluid flows through the Navier-Stokes equations, etc.. Solving these systems efficiently has been a cornerstone of computational physics~\cite{klingenberg_grand_2013}. 

Since the important discovery of an efficient quantum algorithm to solve linear systems~\cite{harrow_quantum_2009}, and subsequent improvements~\cite{ambainis_variable_2010, childs_quantum_2017, wossnig_quantum_2018}, there has been considerable research effort to find algorithms for solving linear ordinary differential equations (ODEs)~\cite{berry_high-order_2014, berry_quantum_2017, costa_quantum_2019, childs_quantum_2020,krovi_improved_2023} and linear partial differential equations (PDEs)~\cite{clader_preconditioned_2013, cao_quantum_2013, montanaro_quantum_2016,engel_quantum_2019,arrazola_quantum_2019, linden_quantum_2022, childs_high-precision_2021}. Recently, there have also been advances in quantum algorithms for nonlinear PDEs~\cite{liu_efficient_2021, xue_quantum_2021, liu_efficient_2023}. Nonlinear PDEs are capable of describing a wide range of dynamics---including, for example, formation of coherent structures, cascades in scale space, deterministic chaos (sensitive dependence on initial conditions), self-organised criticality and turbulence\footnote{Here we distinguish between chaotic and turbulent flows in that while all turbulent flows are chaotic, not all chaotic flows are turbulent: A flow is chaotic if at least one of its Lyapunov exponents is positive, whereby the characteristics of a chaotic flow can be low dimensional. A turbulent flow tends to be a high dimensional chaotic system with a broad band energy spectrum resulting from a cascade of energy across scales, and wherein vorticity plays an important role.}---that underlie the behavior of a variety of complex systems. For modelling increasingly large systems and problem sizes, efficient algorithms are desirable---algorithms that have complexity that scales polynomially in integration time. An early attempt to design quantum algorithms for solving general nonlinear PDEs led to an algorithm whose complexity scales exponentially in integration time~\cite{leyton_quantum_2008}.
In fact, efficient simulation of nonlinear dynamics is generally believed to be not possible as it would allow strict complexity bounds to be violated~\cite{abrams_nonlinear_1998}. In particular, it would allow quantum states to be discriminated with exponentially fewer resources than otherwise possible~\cite{childs_optimal_2016}---a fact we use here to prove the worst-case complexity bound for solving nonlinear differential equations. For these reasons the general algorithm of Ref.~\cite{lloyd_quantum_2020}, based on the simulation of the nonlinear Schrödinger equation, will give exponential time scaling when the nonlinearity is too great. 

The quantum Carleman linearisation (QCL) algorithm~\cite{liu_efficient_2021} uses Carleman embedding~\cite{forets_explicit_2017} to linearise the nonlinear system. Ultimately, quantum mechanics is linear, and solving nonlinear partial differential equations efficiently equates to discretising then linearising a nonlinear differential equation such that it is possible to use a quantum algorithm for solving linear systems of equations~\cite{harrow_quantum_2009, ambainis_variable_2010, childs_quantum_2017, wossnig_quantum_2018}. Rigorous bounds for the regime of an efficient quantum algorithms therefore require that the nonlinear term is not too considerable---the effective ratio of the nonlinear to linear term is defined in the following paragraph.

A nonlinear partial differential equation is discretised, e.g., using the method of finite differences, by taking the continuous variable, say, the function $u(x,t)$ of space and time, $x$ and $t$, and overlaying a finite grid. In one dimension, $x = m\Delta x$ where $m$ is the coordinate and $\Delta x$ is the grid spacing. The variable $u(x,t)$ for a specific $x$ becomes $u_m(t)$, with $\bm{u}(t) = (u_1(t), u_2(t), \dots, u_N(t))$, for $N$ grid points. Discretised derivatives with respect to $x$ essentially use the definition of the derivative between grid points without taking the limit $\Delta x \rightarrow 0$, for example $\frac{\partial u}{\partial x} \rightarrow \frac{u_{i+1} - u_{i-1}}{2 \Delta x}$. 

By adding degrees of freedom (lifting), nonlinear differential equations with a wide variety of nonlinearities, after discretisation using a wide variety of methods, can be written in the form~\cite{hernandez-bermejo_algebraic_1998, gu_qlmor_2011}
\begin{align}
    \frac{d \bm{u}}{dt} = F_2 {\bm{u}}^{\otimes 2} + F_1 \bm{u} + F_0(t),
\end{align}
where the solution vector is $\bm{u} \in \mathbb{R}^N$, with initial condition $\bm{u}(0)$. There is a time-independent nonlinear coefficient, the matrix $F_2 \in \mathbb{R}^{N\times N^2}$, a time-independent linear coefficient, $F_1 \in \mathbb{R}^{N\times N}$, and the time-dependent driven inhomogeneity $F_0(t) \in \mathbb{R}^N$. We consider a restricted set of differential equations where the real part of the eigenvalues of $F_1$ are $\textrm{Re}(\lambda_1) \leq \textrm{Re}(\lambda_2) \leq \dots \leq \textrm{Re}(\lambda_N) < 0$. A Reynolds number-like parameter is also defined that characterises the relative strength of the nonlinearity to the linearity:
\begin{align}
    \label{eq:R_definition}
    R = \frac{1}{\vert \textrm{Re}(\lambda_N) \vert} \left( \Vert F_2 \Vert \Vert \bm{u}(0) \Vert + \frac{\Vert F_0 \Vert}{\Vert \bm{u}(0) \Vert} \right),
\end{align}
where the spectral norm---the largest singular value in the singular value decomposition---is used for $\Vert F_2 \Vert$, the Euclidean norm for $\Vert \bm{u}(0) \Vert$, and the maximum value of the Euclidean norm over the time interval considered is used for $\Vert F_0(t) \Vert$. 

The QCL algorithm uses Carleman linearisation after spatial discretisation. The Carleman embedding treats higher-order terms, i.e. $\bm{u}^{\otimes k}$ for $k \in \mathds{Z}^+$, as independent variables, thereby transforming a finite-dimensional nonlinear system into an infinite dimensional linear system,
\begin{align}
    \frac{d \bm{u}^{\otimes 2}}{d t} &= \bm{u} \otimes \frac{d \bm{u}}{d t} + \frac{d \bm{u}}{d t} \otimes \bm{u} \\
    \frac{d \bm{u}^{\otimes 3}}{d t} &= \bm{u} \otimes \frac{d \bm{u}^{\otimes 2}}{d t} + \frac{d \bm{u}}{d t} \otimes \bm{u}^{\otimes 2} \\
    & ~~\vdots \nonumber
\end{align}
After Carleman linearisation, we have a first order system
\begin{align}
    \frac{d \bm{y}}{d t} = A(t) \bm{y} + \bm{b}(t),
\end{align}
where the infinite dimensional vectors $\bm{y} = (\bm{u}, \bm{u}^{\otimes 2}, \bm{u}^{\otimes 3}, \dots)$ and $\bm{b}(t) = (F_0(t), 0, 0, \dots)$ have been defined. $A(t)$ is an infinite-dimensional matrix that contains $F_2$, $F_1$, and $F_0(t)$ terms in a tridiagonal block structure. This structure mixes vertically adjacent terms in the $\bm{y}$ vector---terms with one more and one fewer product of $\bm{u}$. The linear system is truncated at Carleman truncation order, $C$, giving $\bm{y}_C$, where increasing $C$ decreases the error. The system is then discretised in time on the interval $[0,T]$, via a procedure such as the forward Euler method, with sufficiently small time steps $\delta t$ such that 
\begin{align}
    \bm{y}(\delta t) &\approx \left[ I + A(0)\delta t \right] \bm{y}(0) + \bm{b}(0) \\
    \bm{y}(2 \delta t) &\approx \left[ I + A(\delta t)\delta t \right] \bm{y}(\delta t) + \bm{b}(\delta t) \\
    &~~\vdots \nonumber
\end{align}
Every time step can then be solved simultaneously as the linear system 
\begin{align}
    -(I + \hat{A})\hat{\bm{y}} &= \hat{\bm{b}}
\end{align}
with $\hat{A}$, a large upper bidiagonal block matrix, $\hat{\bm{y}} = (\bm{y}(0), \bm{y}(\delta t), \bm{y}(2\delta t), \dots, \bm{y}(T))$, and $\hat{\bm{b}} = (\bm{b}(0), \bm{b}(\delta t), \bm{b}(2\delta t), \dots, \bm{b}(T))$.
In \cite{liu_efficient_2021} the error of the QCL algorithm was shown to be bounded for arbitrary times provided $R<1$ and solutions are dissipative $\Vert \bm{u}(t) \Vert \leq \Vert \bm{u}(0) \Vert$.
Under the additional assumption that eigenvalues of $F_1$ are negative the QCL algorithm provides an exponential advantage over current classical algorithms in terms of the dimension of the system of equations.

Ref. \cite{liu_efficient_2021} also established worst-case complexity lower bounds for $R\ge \sqrt{2}$,
leaving the region $1 \leq R < \sqrt{2}$ an open question. 
In this paper we close this gap by constructing problems intractable to QCL algorithm with $R\ge 1$.
In particular, we find coupled differential equations that allows state discrimination in integration time $O(\log(1/\varepsilon))$, for initial state overlap $1-\varepsilon$. Since state discrimination has complexity $\Omega(1/\varepsilon)$ any quantum algoithm must have exponential complexity in integration time.

By ``solving'' a PDE using a quantum computer we mean outputting an approximation to a quantum state proportional to the solution vector $\bm{u}(t)$, an amplitude encoded state,
\begin{equation}
\label{eq:amplitude_encoding}
    |u (t)\rangle = \frac{\sum_{i=1}^{n} u_i(t) |i\rangle }{\Vert \bm{u}(t) \Vert}.
\end{equation}
Our results apply to quantum algorithms that aim to solve nonlinear differential equations in this specific way. The QCL algorithm falls into this category. If only a particular property of the flow is required, such as the bulk velocity, it may not be necessary to resolve the small scale structure of the flow. Simulating the mean field solution may be sufficient to extract the desired property, essentially ignoring any turbulence. In fact, extracting this property could be equivalent to solving a simpler non-turbulent flow. We define ``solving" the differential equation as a solution that, with the norm of the solution, would allow any property of the flow to be extracted. For properties related to the turbulence, it may be required to ``solve" the differential equation. The complexity of the QCL algorithm depends on convergence of the Carleman embedding with the truncation order. In general, for $R\geq 1$, truncation order $C$ scales exponentially with time $t$. However, despite the worst-case bound, there are examples of PDEs with $R\ge 1$ that can be linearised with Carleman embedding for arbitrary time and bounded error. Many PDEs with $R \ge 1$ therefore also permit an efficient quantum algorithm. On the other hand, in Section~\ref{sec:positive_lyapunov_no_advantage} we show that for a large class of PDEs no efficient quantum algorithm exists that encodes chaotic solutions in a natural coordinate system. 

\section{\label{sec:tight_lower_bound}Proof of tight lower bound}
Suppose there are two states with large overlap, $1-\varepsilon$, that a black box can generate and we want to determine which of the two states the black box outputs. The overlap of two quantum states cannot be increased through quantum operations. Thus, multiple queries to the black box, or copies of the state, are required to discriminate the two quantum states. 
This state discrimination task cannot be performed in fewer queries than $\Omega(1/\varepsilon$). We give the proof of this in the following Lemma. 
\begin{lemma}
\label{lemma:helstrom_bound}
[Lemma 6 of Ref.~\cite{liu_efficient_2021}] Given a black box that prepares either $|\psi\rangle$ or $|\phi\rangle$, with $\vert \langle \psi | \phi\rangle \vert = 1 - \varepsilon$, then $\Omega(1/\varepsilon)$ queries are required to distinguish the two states. 
\end{lemma}
\begin{proof}
The overlap of two states $|\psi\rangle$ and $|\phi\rangle$ is $\vert \langle \psi | \phi\rangle \vert = 1 - \varepsilon$. With $k$ queries to the black box, oracle $O_s$, we can prepare either state $|\psi\rangle^{\otimes k}$ or $|\phi\rangle^{\otimes k}$. The overlap of these states is $(1-\varepsilon)^k$. For pure states the trace distance is $d = \Vert |\psi\rangle^{\otimes k}  - |\phi\rangle^{\otimes k}  \Vert_\textrm{tr} = \sqrt{1-\vert \langle \psi | \phi \rangle \vert^{2k}}=\sqrt{1-(1-\varepsilon)^{2k}}$ and satisfies $d \leq \sqrt{2k\varepsilon}$.
At least $k = \Omega(1/\varepsilon)$ queries are required to achieve $O(1)$ distance $d$, and therefore a bounded-error separation between the states for distinguishability. The Helstrom bound~\cite{helstrom_quantum_1969} states that the trace distance can be used to bound the success probability of any state discrimination protocol.
\end{proof}

First, we assume---and we give an explicit example in the following Theorem~\ref{theorem:R_greater_than_1}---that there are pairs of solutions to nonlinear differential equations that when simulated as a quantum state would lead to two states with initial overlap of $1-\varepsilon$ being reduced to an overlap of $1-O(1)$ in $O(\log(1/\varepsilon))$ simulated time. The result of Lemma~\ref{lemma:helstrom_bound} leads to the conclusion that the number of copies of the initial state needed by this quantum algorithm is exponential in the simulation time. In the following Lemma we formalise and prove this statement.

\begin{lemma}
\label{lemma:time_to_reduce_overlap}
A black box generates one of two states $|\psi_0\rangle$ or $|\phi_0\rangle$ with overlap $\vert \langle \psi_0 | \phi_0 \rangle \vert = 1-\varepsilon$ with $0 < \varepsilon < e^{-4} \approx 0.018$. If the overlap of these two states under evolution that simulates the solution of a differential equation is reduced to $\vert \langle \psi(T) | \phi(T) \rangle\vert \leq 0.95$ in time $T = O(\log(1/\varepsilon))$, then the number of queries to the black box needed to solve the differential equation is exponential in $T$.
\end{lemma}
\begin{proof}
By assumption, the trace distance between the two states is
\begin{align}
    \Vert |\psi(T)\rangle - |\phi(T)\rangle\Vert_\textrm{tr} &= \sqrt{1-|\langle \psi(T)| \phi(T)\rangle|^2} \\
    &\geq \frac{\sqrt{39}}{20}
\end{align}
after simulated time $T = O(\log(1/\varepsilon))$. 
As trace distance is non-increasing under quantum operations, sufficient number of copies of the two states must be provided to the quantum algorithm. From $\Vert |\psi(0)\rangle^{\otimes k} - |\phi(0)\rangle^{\otimes k}\Vert_{\textrm{tr}} \le \sqrt{2 k \varepsilon}$, we require $k > 39/(800\varepsilon)$ queries to give a separation of $\sqrt{39}/20$. Since $\varepsilon$ is exponentially small in integration time $T$, simulating up to $T$ requires a number of queries exponentially large in $T$.
\end{proof}

In the following, we use Lemmas~\ref{lemma:helstrom_bound} and~\ref{lemma:time_to_reduce_overlap} to prove that for $R \geq 1$ there is a nonlinear differential equation that cannot be solved efficiently by any quantum algorithm.
Unlike the proof of Lemma 7 in Ref.~\cite{liu_efficient_2021}, which uses uncoupled differential equations, in Theorem~\ref{theorem:R_greater_than_1} we use two coupled differential equations to prove the worst-case lower bound. 
\begin{theorem}
\label{theorem:R_greater_than_1}
There is a nonlinear ODE with $R\geq 1$, as defined in Eq.~\eqref{eq:R_definition}, such that any bounded-error quantum algorithm that gives a quantum state with amplitudes proportional to the solution vector must have complexity exponential in integration time. 
\end{theorem}
\begin{proof}
We define a coupled differential equation with solution vector $\bm{u}(t) = (u_1(t), u_2(t))^\intercal$,
\begin{align}
    \label{eq:1}
    \frac{d u_1}{dt} &= - u_1 + R^* u_1^2 \\
    \label{eq:2}
    \frac{d u_2}{dt} &= - 2 u_2 - u_1.
\end{align}
The general solution is 
\begin{align}
    \label{eq:general_solution}
    \bm{u} (t) = \begin{pmatrix}
        \left[R^* + e^{t+c_1}\right]^{-1} \\
        e^{-2 (t+c_1)} \left(e^{2 c_1}c_2 - e^{t+c_1} +R^*\log(R^* + e^{t+c_1})\right)
    \end{pmatrix},
\end{align}
with constants $c_1$ and $c_2$ determined by the initial conditions.
For the initial conditions $u_1(0) = 1-\varepsilon$ and $u_2(0) = \delta= \sqrt{2\varepsilon - \varepsilon^2}$, with $0 \leq \varepsilon \ll 1$, and $\Vert \bm{u} \Vert = 1 $, we have 
\begin{equation}
\label{eq:general_solution_bcs}
    \bm{u} (t) = \begin{pmatrix}
        \left[R^* + e^{t} \frac{1-R^* (1-\varepsilon)}{1-\varepsilon}\right]^{-1} \\
        e^{-2t} \left( \frac{1 + \delta - R^* \delta - e^t(1-\varepsilon) - \varepsilon + R^* \delta\varepsilon}{1-R^* (1-\varepsilon)} + \frac{R^* (1-\varepsilon)^2 \log(R^*(1-\varepsilon)+e^t(1-R^* (1-\varepsilon)))}{{(1-R^* (1-\varepsilon))^2}}\right)
    \end{pmatrix}.
\end{equation}
From the definition of $R$ in Eq.~\eqref{eq:R_definition}, it follows that $R = R^*$. Consider the boundary case $R=1$, we have the analytic solution
\begin{equation}
\label{eq:u_evolution_epsilon}
    \bm{u}^\psi (t) = \begin{pmatrix}
        \left[1+\frac{\varepsilon e^t}{1-\varepsilon}\right]^{-1} \\
        \frac{e^{-2t}}{\varepsilon^2}\left[ \varepsilon(1-e^t + \varepsilon e^t -\varepsilon + \delta\varepsilon) + (1-\varepsilon)^2\log(1 -\varepsilon + e^{t}\varepsilon)\right]
    \end{pmatrix}.
\end{equation}
For $\varepsilon = 0$, so $\delta=0$, we define the solution $\bm{u}^\phi(t)$, and the second element $u_2^\phi(t)$ gives
\begin{align}
    u_2^\phi(t) &= \lim_{\varepsilon\rightarrow 0} \frac{e^{-2t}}{\varepsilon^2}\left[\varepsilon(1 - e^t +\varepsilon e^t - \varepsilon + \delta\varepsilon) + (1-\varepsilon)^2\log(1 -\varepsilon + e^{t}\varepsilon)\right] \\
    &= \lim_{\varepsilon\rightarrow 0} \frac{e^{-2t} \frac{d}{d\varepsilon} \left[ \varepsilon(1 - e^t +\varepsilon e^t - \varepsilon + \delta\varepsilon) + (1-\varepsilon)^2\log(1 -\varepsilon + e^{t}\varepsilon)\right] }{\frac{d}{d\varepsilon}\varepsilon^2} \\
    &= \lim_{\varepsilon\rightarrow 0} \frac{e^{-2t}}{2\varepsilon} \Bigg[ 1 + \varepsilon e^t - e^t  - \varepsilon + \delta\varepsilon +\varepsilon \left( e^t  - 1 + \delta + \frac{(1-\varepsilon)\varepsilon}{\delta}\right) \nonumber \\ & \hspace{4.5cm}  -2(1-\varepsilon)\log(1-\varepsilon+e^t\varepsilon) - \frac{(1-\varepsilon)^2 (1 - e^t)}{1-\varepsilon+e^t \varepsilon}  \Bigg]  \\
    &= \lim_{\varepsilon\rightarrow 0} \frac{e^{-2t} \left[ \varepsilon e^t  - \varepsilon +\varepsilon \left( e^t  - 1 \right) - 2\log(1-\varepsilon+e^t\varepsilon) + 1 - e^t  - \frac{(1-\varepsilon)^2 (1 - e^t)}{1-\varepsilon+e^t \varepsilon}  \right] }{2\varepsilon} \\
    &= \frac{e^{-2t}}{2} \Bigg\{ 2 \left( e^t - 1 \right) -  \lim_{\varepsilon\rightarrow 0} \frac{2\log(1-\varepsilon+e^t\varepsilon)}{\varepsilon} \nonumber \\ 
    &\hspace{5.5cm}  +  \lim_{\varepsilon\rightarrow 0} \left[ \frac{1}{\varepsilon}\left(1 - e^t - \frac{(1-\varepsilon)^2 (1 - e^t)}{1-\varepsilon+e^t \varepsilon} \right) \right] \Bigg\} \\
    &= \frac{e^{-2t}}{2} \Bigg\{ 2 \left( e^t - 1 \right) -  \lim_{\varepsilon\rightarrow 0} \frac{2 \frac{d}{d\varepsilon} \log(1-\varepsilon+e^t \varepsilon)}{\frac{d}{d\varepsilon} \varepsilon} \nonumber\\& \hspace{3.5cm}+   \lim_{\varepsilon\rightarrow 0} \left[ \frac{1}{\varepsilon}\left(1 - e^t - (1-\varepsilon)^2 (1 - e^t)(1+\varepsilon-e^t \varepsilon) \right) \right] \Bigg\} \\
    &=  \frac{e^{-2t}}{2} \left\{ 2\left(  e^t -1 \right) -  2 (e^t-1)  +  \lim_{\varepsilon\rightarrow 0} \left[ \frac{1}{\varepsilon}\left(1 - e^t - 1 + e^t + \varepsilon - e^{2t} \varepsilon \right) \right] \right\}  \\
     &=  \frac{e^{-2t}}{2} \left\{ \lim_{\varepsilon\rightarrow 0} \left[ \frac{1}{\varepsilon} \left( \varepsilon - e^{2t}\varepsilon \right) \right] \right\}  \\
    &= \frac{e^{-2t} - 1}{2},
\end{align}
where we have used l'Hôpital's rule.
We therefore have a solution
\begin{equation}
    \label{eq:u_evolution_0}
    \bm{u}^\phi (t) = \begin{pmatrix}
        1 \\
        \frac{1}{2}(e^{-2t} - 1)
    \end{pmatrix},
\end{equation}
for $\varepsilon = 0$, and a solution $\bm{u}^\psi (t)$ for small positive $\varepsilon$. 
These solutions, with entries $u^{\psi}_1$, $u^\phi_1$, and $u^\psi_2$, $u^\phi_2$, define two quantum states at initial time, $t=0$: $|\psi_0\rangle = (1-\varepsilon)|0\rangle + \delta|1\rangle$ and $|\phi_0 \rangle = |0\rangle$. The initial overlap is $\vert \langle \psi_0|\phi_0 \rangle \vert = 1-\varepsilon$. We define an oracle $O_t$ that updates the solution to time $t$, such that the states evolve according to the solutions of the differential equations up to normalisation. We further define ratios of the solutions $S_{\psi}(t) = \frac{u^{\psi}_{1}}{u^{\psi}_{2}}$ and $S_{\phi}(t) = \frac{u^{\phi}_{1}}{u^{\phi}_{2}} $. 
In particular, 
\begin{align}
    \label{eq:S_phi_expression}
    S_{\phi}(t) = \frac{-2}{1-e^{-2t}},
\end{align}
so $S_{\phi}(t) \leq -2$ and $\lim_{t\rightarrow\infty} S_
\phi (t) = -2$. 
After time $t$, the initial state is evolved by the oracle $|\phi(t)\rangle = O_t|\phi_0\rangle$ such that we have the state
\begin{align}
    |\phi(t)\rangle &= \frac{1}{\sqrt{1 + S_\phi (t)^2}}(S_\phi(t)|0\rangle + |1\rangle).
\end{align} 
We now compute $S_\psi(t)$ and its scaling with $\varepsilon$,
\begin{align}
    S_\psi(t) &= \frac{ \left[1+\frac{\varepsilon e^t}{1-\varepsilon}\right]^{-1} }{\frac{e^{-2t}}{\varepsilon^2}\left[ \varepsilon(1-e^t + \varepsilon e^t -\varepsilon + \delta\varepsilon) + (1-\varepsilon)^2\log(1 -\varepsilon + e^{t}\varepsilon)\right]}. \\
    &= \frac{e^{2t}(1-\varepsilon)\varepsilon^2}{(1+(-1+e^t)\varepsilon)\left[ \varepsilon(1-e^t + \varepsilon e^t -\varepsilon + \delta\varepsilon) + (1-\varepsilon)^2\log(1 -\varepsilon + e^{t}\varepsilon)\right]}
\end{align}
The ratio $S_\psi (t)$ tends to $-1$, $\lim_{t\rightarrow\infty}S_\psi (t) = -1$. The numerator $e^{2t}(1-\varepsilon)\varepsilon^2$ is positive for all time, as is the term $(1+(-1+e^t)\varepsilon)$ for $t>0$. The term $\varepsilon(1-e^t + \varepsilon e^t -\varepsilon + \delta\varepsilon)$ is negative for $t > \log(1 + \frac{\delta\varepsilon}{1-\varepsilon})$. The term is therefore negative for $t \geq \log(1/\varepsilon)$ if $\varepsilon\leq \frac{1}{2}$, and it monotonically decreases exponentially. The term $(1-\varepsilon)^2\log(1 -\varepsilon + e^{t}\varepsilon)$ is positive for $t>0$ and monotonically increases sub-exponentially---it grows linearly. Hence, by substituting in $t = \log(1/\varepsilon)$, we have that for time $t \geq \log(1/\varepsilon)$, the ratio $S_\psi (t)$ is negative (in fact, the ratio is negative much earlier than this time, but we do not have to consider earlier times). 
The absolute value of the ratio for $t \geq \log(1/\varepsilon)$ and $\varepsilon\leq \frac{1}{2}$ is therefore 
\begin{align}
    \vert  S_\psi(t) \vert = - S_\psi(t).
\end{align}
The ratio can be bounded for $t>\log(1/\varepsilon)$ by
\begin{align}
    |S_\psi (t)| &\leq \frac{e^{2t}(1-\varepsilon)\varepsilon^2}{(1+(-1+e^t)\varepsilon)\left[ \varepsilon(1-e^t + \varepsilon e^t) + (1-\varepsilon)^2\log(1 -\varepsilon + e^{t}\varepsilon)\right]} \\
    &\leq \frac{1}{1+\frac{ e^{-t} (1-3 \varepsilon +\varepsilon^2) -  e^{-2 t} (1-\varepsilon ) }{ \varepsilon(1-\varepsilon ) }-\frac{ e^{-2 t} (1-\varepsilon ) \left(\left(e^t-1\right) \varepsilon +1\right) }{  \varepsilon ^2} \log(1 -\varepsilon + e^{t}\varepsilon)}.
\end{align}
The denominator tends to 1 from below. The term $a(t) = \frac{ e^{-t} (1-3 \varepsilon +\varepsilon^2) -  e^{-2 t} (1-\varepsilon ) }{ \varepsilon(1-\varepsilon ) }$ is positive after time $t > \log(\frac{1-\varepsilon}{1-3\varepsilon+\varepsilon^2})$, which means the term is positive at times $t > \log(1/\varepsilon)$ for $\varepsilon<1-\frac{1}{\sqrt{2}}$. The term $b(t) = -\frac{ e^{-2 t} (1-\varepsilon ) \left(\left(e^t-1\right) \varepsilon +1\right) }{  \varepsilon ^2} \log(1 -\varepsilon + e^{t}\varepsilon)$ is negative for $t>0$ and monotonically tends to $0$. The denominator can thus be bounded from below for $\varepsilon<1-\frac{1}{\sqrt{2}}$ and $t > \log(1/\varepsilon)$ by $1+b(t)$. However, the absolute value of the ratio $\vert S_\psi (t) \vert$ will only be upper bounded when $1 + b(t) > 0$. At time $t = 2 \log(1/\varepsilon)$, we find 
\begin{align}
    b(t) &= -\varepsilon  \left((2-\varepsilon) \varepsilon ^2+1\right) \log\left(\frac{1}{\varepsilon }+1-\varepsilon\right) \\
    &> -\frac{1}{2} \label{eq:b_t_bound}
\end{align}
for $\varepsilon<1-\frac{1}{\sqrt{2}}$. 
Finally, this gives a bound on the ratio for $t> 2\log(1/\varepsilon)$ and $\varepsilon<1-\frac{1}{\sqrt{2}}$ of
\begin{align}
    |S_\psi (t)| &\leq \frac{1}{1-\frac{ e^{-2 t} (1-\varepsilon ) \left(\left(e^t-1\right) \varepsilon +1\right) }{ \varepsilon ^2} \log(1 -\varepsilon + e^{t}\varepsilon)} \\
    \label{eq:S_psi_bound}
    &\leq \frac{1}{1-\frac{t e^{- t}  }{\varepsilon}},
\end{align}
where we have used $b(t) \geq -\frac{t e^{-t}}{\varepsilon}$. Fig.~\ref{fig:S_psi_S_phi} shows a comparison of ratios $S_{\psi}(t)$ and $S_{\phi}(t)$ for various $\varepsilon$.
\begin{figure}[t]
    \centering
    \includegraphics[scale=0.46]{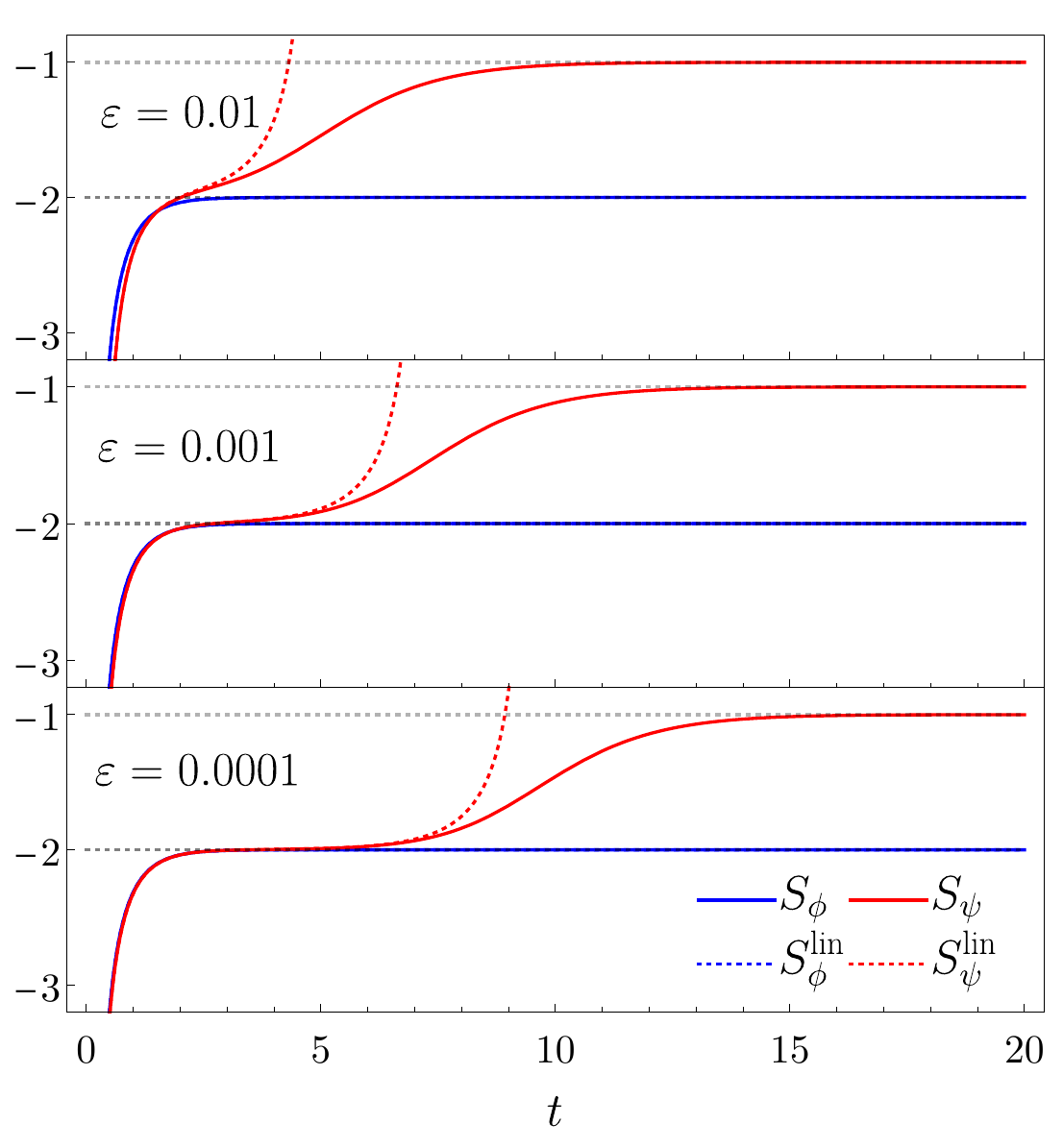}
    \caption{Comparison of ratios $S_{\psi}$ (red solid line) and $S_{\phi}$ (blue solid line) for $\varepsilon=0.01, 0.001, 0.0001$ up to time $t=20$. The behaviour for each $\varepsilon$ is similar. $S_{\phi}$ tends to $-2$ asymptotically and $S_{\psi}$ tends to $-1$. The dotted lines are the ratios of the solutions for the linearised equations around the initial condition. It can be seen that the time where the two lines start diverging scales logarithmically in $1/\varepsilon$ as proven in the text. The dotted lines show the solution to the linearised equations. $S_\phi$ and $S_\phi^\text{lin}$ exactly overlap.} 
    \label{fig:S_psi_S_phi}
\end{figure}
Similarly to $|\phi(t)\rangle$, after time $t$ with the initial condition of $|\psi_0\rangle$, we have the state 
\begin{align}
    |\psi(t) \rangle = \frac{1}{\sqrt{1+S_\psi(t)^2}}(S_\psi (t) |0\rangle + |1\rangle).
\end{align}
The magnitude of the overlap at time $t>2\log(1/\varepsilon)$, where both $S_\phi(t)$ and $S_\psi(t)$ are negative and real, is therefore
\begin{align}
    |\langle \psi (t)|\phi (t)\rangle| &= \frac{|S_\phi (t)| |S_\psi (t)| + 1}{\sqrt{S_\phi(t)^2 + 1}\sqrt{S_\psi(t)^2+1}}.
\end{align} 
Using Eq.~\eqref{eq:b_t_bound}, we find $|S_\psi(t)| < |S_\phi(t)|$ for $\varepsilon < 1 - \frac{1}{\sqrt{2}}$ and $t>2\log(1/\varepsilon)$, giving
\begin{align}
    \frac{\partial }{\partial \vert S_\psi (t) \vert}  \left(|\langle \psi (t)|\phi (t)\rangle| \right) &= \frac{\vert S_\phi (t) \vert-\vert S_\psi (t) \vert}{\sqrt{\vert S_\phi (t) \vert^2 + 1} \left(\vert S_\psi (t) \vert^2+1 \right)^{\frac{3}{2}}}, \\
    &> 0. \label{eq:overlap_derivative_sign}
\end{align}
The overlap $|\langle \psi (t)|\phi (t)\rangle|$ increases with increasing $|S_\psi (t)|$.
Hence, the upper bound of $|S_\psi(t)|$ gives an upper bound for the overlap. So taking $t^* = 3\log(1/\varepsilon)$, and using Eqs.~\eqref{eq:S_phi_expression} and~\eqref{eq:S_psi_bound} we have
\begin{align}
    |\langle \psi (t^*)|\phi (t^*)\rangle| \leq \frac{2 + (1-\varepsilon^6)(1-3\varepsilon^2 \log(1/\varepsilon)))}{\sqrt{4+(1-\varepsilon^6)^2} \sqrt{1+(1-3\varepsilon^2 \log(1/\varepsilon))^2}}.
\end{align}
For $0 < \varepsilon < e^{-4} \approx 0.018$, we therefore have the bound 
\begin{align}
    |\langle \psi (t^*)|\phi (t^*)\rangle| < 0.95
\end{align}
Using Lemma~\ref{lemma:time_to_reduce_overlap}, it is not possible to output these states with only polynomial queries---polynomial time.
There cannot be a polynomial time quantum algorithm to solve this coupled nonlinear differential equation for $R=1$. The same arguments apply for $R^* \geq 1$, and therefore $R \geq 1$. We can choose new initial conditions such that we have 
\begin{align}
    \bm{u}^{\phi} (0) = \left(\frac{1}{R}, \sqrt{1-\left(\tfrac{1}{R}\right)^2} \right)^\intercal,
\end{align}
and 
\begin{align}
    \bm{u}^{\psi} (0) = \left(\frac{1-\varepsilon_*}{R}, \sqrt{1-\left(\frac{1-\varepsilon_*}{R}\right)^2} \right)^\intercal,
\end{align}
where $\varepsilon_*$ is chosen such that the initial overlap is at least $1-\varepsilon$. If $\varepsilon_* = \varepsilon + \sqrt{\varepsilon(2-\varepsilon)(R^2-1)}$, then the initial overlap is $1-\varepsilon$, and the choice $\varepsilon_* = \frac{\varepsilon}{R}$ gives an initial overlap $|\langle \psi (0)|\phi (0) \rangle | \geq 1-\varepsilon$. The solutions become 
\begin{align}
    \bm{u}^\phi (t) = \begin{pmatrix}
        \frac{1}{R} \\
        \frac{1}{2R}\left(e^{-2t}\left(R+2\sqrt{R^2-1}\right)-R \right)
    \end{pmatrix},
\end{align}
and 
\begin{align}
    \bm{u}^\psi (t) = \begin{pmatrix}
        \frac{1}{R}\left[1+\frac{\varepsilon_* e^t}{1-\varepsilon_*}\right]^{-1} \\
        \frac{1}{R}\frac{e^{-2t}}{\varepsilon_*^2}\left[ \varepsilon_* (1-e^t + \varepsilon_* e^t -\varepsilon_* + \delta_*\varepsilon_*) + (1-\varepsilon_*)^2\log(1 -\varepsilon_* + e^{t}\varepsilon_*)\right]
    \end{pmatrix},
\end{align}
where $\delta_* = \sqrt{R^2 - (1-\varepsilon_*)^2} \leq R$ . Using similar arguments as the $R=1$ case, but with $\varepsilon_* = \frac{\varepsilon}{R}$, leads to the bound 
\begin{align}
    \vert S_{\psi} (t) \vert \leq \frac{1}{1-\frac{R t e^{-t}}{\varepsilon}},
\end{align}
for $t > 2 \log(R/\varepsilon)$. For the ratio, $\vert S_\phi (t)\vert$, we have 
\begin{align}
    \vert S_\phi (t)\vert &=  \frac{2}{1-(1+2 \sqrt{R^2-1})e^{-2 t} } \\
    &\geq \frac{2}{1- e^{-2 t} }\label{eq:S_phi_R_bound}.
\end{align}
As in the $R=1$ case, we have the ratio of the solutions $\vert S_{\psi}(t)\vert  < \vert S_{\phi}(t)\vert$ for $t > 2 \log(R/\varepsilon)$, which means, using an expression equivalent to Eq.~\eqref{eq:overlap_derivative_sign}, the bound of Eq.~\eqref{eq:S_phi_R_bound} can be used 
to give an upper bound on the overlap at $t^* = 3 \log(R/\varepsilon)$ of
\begin{align}
    |\langle \psi (t^*)|\phi (t^*)\rangle| \leq \frac{2 + (1-\varepsilon_*^6)(1-3\varepsilon_*^2 \log(1/\varepsilon_*)))}{\sqrt{4+(1-\varepsilon_*^6)^2} \sqrt{1+(1-3\varepsilon_*^2 \log(1/\varepsilon_*))^2}}.
\end{align}
Since $\varepsilon_* = \frac{\varepsilon}{R} < \varepsilon$, we find 
\begin{align}
    |\langle \psi (t^*)|\phi (t^*)\rangle| < 0.95
\end{align}
for $0 < \varepsilon < e^{-4} \approx 0.018$. As in the case of $R=1$, by using Lemma~\ref{lemma:time_to_reduce_overlap}, we can therefore conclude that there cannot be a polynomial time quantum algorithm to solve this coupled nonlinear differential equation for $R\geq 1$.
\end{proof}

It is insightful to analyse the coupled differential equations used in the proof of Theorem~\ref{theorem:R_greater_than_1} in the linearised regime. Letting $u_1(t) = 1+ x_1(t)$ and $u_2(t) = x_2(t)$ in Eqs.~\eqref{eq:1} and \eqref{eq:2} and discarding nonlinear terms in $x_1$ and $x_2$ we obtain
\begin{align}
    \dot{x}_1 &= x_1 \\
    \dot{x}_2 &= -x_1 -2 x_2 -1
\end{align}
In Fig.~\ref{fig:S_psi_S_phi} we plot $S_\phi^\text{lin}$ and $S_\psi^\text{lin}$ which are defined as before but using the solutions to the linearised equations above. 
$S_\phi^\text{lin}$ happens to be exactly equal to $S_\phi$ for all times. On the other hand, $S_\psi^\text{lin}$ closely follows $S_\psi$ up to a certain time before diverging from it. 
The important thing to note is that $S_\phi$ and $S_\psi$ already differ by a constant amount in a region where the linearised dynamics is still a good approximation. 
This suggests that simulating the corresponding linearised ODE with a quantum computer should also be inefficient. 
This is confirmed by analysing the
eigenvalues of the corresponding $F_1$ matrix, which are $\{-2, 1\}$. Quantum algorithms for linear ODEs assume the real part of the eigenvalues of $F_1$ are nonpositive~\cite{berry_high-order_2014, berry_quantum_2017, childs_quantum_2020} because the dynamics cannot be simulated efficiently if the eigenvalues are positive. This is due to the fact that there are always pairs of trajectories that diverge from each other exponentially quickly. The two solutions of the linearised equations in Fig.~\ref{fig:S_psi_S_phi} are an example of this. 

The linearised approximations of nonlinear systems can show trajectories that diverge rapidly and indicate a complexity bound. In a similar way, in the next section, we study the complexity of quantum algorithms for simulating chaotic and turbulent systems by the Lyapunov exponent, which characterises the rate of separation of close trajectories.

\section{\label{sec:positive_lyapunov_no_advantage}No efficient quantum algorithm for positive Lyapunov exponents}
In this section, we investigate the limitations of quantum algorithms to simulate dynamics of chaotic systems. 

The 2-norm for a matrix $A$ is defined with respect to the standard Euclidean vector norm, 
\begin{align}
    \label{eq:def_matrix_2_norm}
    \Vert A \Vert_2 = \sup_{\bm{v}\ne 0} \frac{ \Vert A \bm{v} \Vert}{\Vert \bm{v}\Vert}.
\end{align}
We define a point in phase space $\bm{u}(t)$ at some time $t$. Trajectories in phase space are determined by solving the general differential equation 
\begin{align}
    \frac{d\bm{u}}{dt} = f(\bm{u}(t),t),
\end{align}
with initial condition at $t=0$. Let the propagator $\bm{F}_\tau (\bm{u})$ update the trajectories by time $\tau$,
\begin{align}
    \bm{F}_\tau(\bm{u}(0)) = \bm{u}(\tau).
\end{align}
Consider initially infinitesimal separations between two points, $\bm{\delta u}(0)$. The change over time of the separation of two points in phase space is the ratio of the Euclidean norm at a time $\tau$ to the norm of the initial separation,
\begin{align}
    \frac{\Vert \bm{\delta u}(\tau) \Vert}{\Vert \bm{\delta u}(0) \Vert} &=  \frac{\Vert \bm{F}_\tau (\bm{u}(0)+\bm{\delta u}(0)) - \bm{F}_\tau(\bm{u}(0))\Vert}{\Vert \bm{\delta u}(0)\Vert}.
\end{align}
The Lyapunov exponent can now be defined~\cite{kuznetsov_stability_2005}.
\begin{definition}
\label{def:ftle}
The Lyapunov exponent is
\begin{align}
    \lambda(\bm{u}(0)) &= \lim_{\tau\rightarrow\infty}\frac{1}{\tau} \log \left( \sup_{\bm{\delta u}(0) \rightarrow 0} \frac{\Vert \bm{\delta u}(\tau) \Vert}{\Vert \bm{\delta u}(0)\Vert} \right) 
\end{align}
where the supremum is taken over all possible directions of $\bm{\delta u}(0)$, and, using the definition of matrix 2-norm in Eq.~\eqref{eq:def_matrix_2_norm}, gives 
\begin{align}
    \label{eq:lyapunov_exponent_definition}
    \lambda(\bm{u}(0)) = \lim_{\tau\rightarrow\infty}\frac{1}{\tau} \log \left( \Vert \nabla \bm{F}_\tau (\bm{u}(0)) \Vert_2 \right).
\end{align}
\end{definition}

The chaotic regime extends to some boundary. This boundary encloses a region of phase space called the chaotic attractor, within which we find chaotic solutions. Typically the basin of the attractor encompasses a significant volume of phase space. Trajectories within the attractor basin cannot leave. It is possible that the attractor is displaced from the coordinate origin.

Consider two trajectories in the chaotic regime with a positive Lyapunov exponent that are initially close to each other. 
Although the phase space points quickly diverge from each other, the distance between the quantum states that encode these points may or may not increase. 
This is because the quantum states are normalized whereas phase space points are not. For example, trajectories can separate if the size of the phase space points grow at different rates, without changing the angle between them. The corresponding quantum states, on the other hand, remain close. The amplitude encoding of Eq.~\eqref{eq:amplitude_encoding} is not useful for discriminating such trajectories and depending on the purpose of the simulation may be inadequate.

A similar situation occurs when the initial phase space points one wants to simulate are far away from the origin. The angular separation for normalised quantum states corresponding to different phase space points decreases with the distance from origin along with ones ability to distinguish the quantum states. As a simple example, consider an attractor as a sphere of radius $s$ and a distance $a_0$ from the origin. Two quantum states in the attractor can only have a maximum angular separation of $\theta_{\textrm{max}} = 2\arctan(s/a_0)$. Regardless of whether the solution is chaotic or not, the number of copies of quantum states required to distinguish two quantum states encoding two nearby trajectories (and hence obtain some information about the dynamics) increases as the centre of the attractor is displaced from the origin.
We therefore consider that placing the origin such that it is surrounded by the phase space region of interest to be the natural coordinate system choice.

With the definition of Lyapunov exponent, which is illustrated in Fig.~\ref{fig:phase_space_plot}, we can make the following statement about solutions of nonlinear differential equations with positive Lyapunov exponents.
\begin{figure}[h]
    \centering
    \includegraphics[scale=0.85]{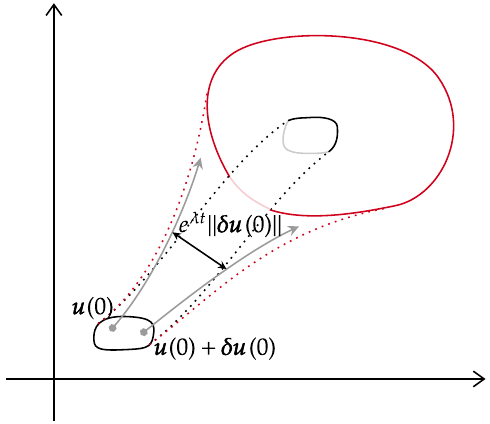}
    \caption{Phase space depiction of a region where two nearby initial conditions diverge exponentially (red) and a regime that is conservative and nearby initial conditions remain close together (black).}
    \label{fig:phase_space_plot}
\end{figure}
\begin{theorem}
\label{theorem:no_algo_for_choas}
Consider a nonlinear differential equation with a coordinate axis such that the origin is surrounded by a chaotic attractor with positive Lyapunov exponent and solution norms that grow sub-exponentially. Any quantum algorithm that produces a quantum state with amplitudes proportional to the solution vector of trajectories within the attractor has complexity that scales exponentially in time.
\end{theorem}
\begin{proof}
The distance between two classical points on the phase space is $\Vert \bm{u}^A(t) - \bm{u}^B(t) \Vert = \Vert \bm{\delta u}(t) \Vert$, for points $\bm{u}^A(t), \bm{u}^B(t) \in \mathbb{R}^N$ with $N = 2^n$ dimensions, elements 
\begin{align}
    \bm{u}^{A}(t) = (u^{A}_1(t), u^{A}_2(t), \dots, u^{A}_N(t))^\intercal 
\end{align}
and 
\begin{align}
    \bm{u}^{B}(t) = (u^{B}_1(t), u^{B}_2(t), \dots, u^{B}_N(t))^\intercal,
\end{align}
and separation $\bm{\delta u}(t) = \bm{u}^{A}(t) - \bm{u}^{B}(t)$. The norm of the phase space is chosen such that initially we have $\Vert \bm{u}^A(0) \Vert = \Vert \bm{u}^B(0) \Vert = 1$. This is possible by scaling the differential equations and choosing a $\bm{u}^B(0)$ that is only a small angle from $\bm{u}^A(0)$ with respect to the origin. 

The phase space points can be described by the amplitudes of the quantum states $|u^A(t)\rangle$ and $|u^B(t)\rangle$. 
As in the introduction, we assume the amplitude encoding for some orthogonal basis
\begin{equation}
    |u^{A} (t)\rangle = \frac{\sum_{i=1}^{n} u^{A}_i(t) |i\rangle }{\Vert u^{A}(t) \Vert},
\end{equation}
and 
\begin{equation}
    |u^{B} (t)\rangle = \frac{\sum_{i=1}^{n} u^{B}_i(t) |i\rangle }{\Vert u^{B}(t) \Vert}.
\end{equation}
For pure states, the overlap between the two quantum states in terms of the phase space points is 
\begin{align}
    \langle u^A(t) | u^B(t) \rangle &= \frac{\bm{u}^A(t) \cdot \bm{u}^B(t)}{\Vert \bm{u}^A(t) \Vert \Vert \bm{u}^B(t) \Vert} \\
    &= \cos(\theta(t)), \label{eq:overlap_t}
\end{align}
where $\theta(t)$ is the angle between the phase space vectors.

Initially, the phase space vectors are at a separation $\Vert \bm{\delta u}(0) \Vert$ and are a small angle $\theta_0 = \theta(0) \ll \pi/2$ apart. Since the initial magnitudes of $\bm{u}^A(0)$ and $\bm{u}^B(0)$ are unity, the initial separation is 
\begin{align}
    \Vert \bm{\delta u}(0)\Vert^2 &= \| |u^A(0)\rangle-|u^B(0)\rangle\|^2\\ 
    &= 2-2 \langle u^A(0)|u^B(0)\rangle \\
    &= 2-2\cos(\theta_0).
    \label{eq:initial_separation}
\end{align}
Consider a flow (continuous-time evolution) with a positive Lyapunov exponent. Initial points must separate exponentially. In practice, however, we use a finite time Lyapunov exponent $\lambda_t$. After a time $t$, the separation of the points will be of the order of the size of the attractor and no longer increases exponentially. The Lyapunov exponent can therefore be written as
\begin{align}
    \lambda_t = \frac{1}{t} \log\left(\frac{\Vert \bm{\delta u}(t) \Vert}{\Vert \bm{\delta u}(0) \Vert} \right) > 0,
\end{align}
for large enough $t$---when $\Vert \bm{\delta u}(t) \Vert$ approaches the size of the attractor---and we have $\lambda_t \approx \lambda$ of Def.~\ref{def:ftle}. Within the linearised regime for initial times and trajectories, $\Vert \bm{\delta u}(t) \Vert = e^{\lambda_t t} \Vert \bm{\delta u}(0)\Vert$. 
We have 
\begin{align}
    \Vert \bm{\delta u}(t)\Vert^2 &= \Vert \bm{u}^A(t) \Vert^2 + \Vert  \bm{u}^B(t) \Vert^2 - 2 \bm{u}^A(t) \cdot \bm{u}^B(t) \\
    &= \Vert \bm{u}^A(t) \Vert^2 + \Vert  \bm{u}^B(t) \Vert^2 - 2 \langle u^A(t)|u^B(t) \rangle \Vert\bm{u}^A(t) \Vert \Vert\bm{u}^B(t) \Vert.
\end{align}
The quantum state overlap is therefore
\begin{align}
    \langle u^A(t)|u^B(t) \rangle &= \frac{\Vert \bm{u}^A(t) \Vert^2 + \Vert  \bm{u}^B(t) \Vert^2 - e^{2\lambda_t t} \Vert \bm{\delta u}(0)\Vert^{2}}{2 \Vert \bm{u}^A(t) \Vert \Vert \bm{u}^B(t) \Vert}.
\end{align}
For the nonlinear differential equation we are considering, the norms $\Vert \bm{u}^{A}(t) \Vert$ and $\Vert \bm{u}^{B}(t)\Vert$ grow sub-exponentially.
Sub-exponential growth is defined such that the norm can be bounded by 
\begin{align}
    \Vert \bm{u}^A(t)\Vert \leq e^{t^{1-\alpha}},
\end{align}
where $\alpha >0$ is small, positive, and real. The ratio of the norm of the solution to the norm of the separation is therefore bounded by
\begin{align}
    \frac{\Vert \bm{u}^A(t)\Vert}{\Vert\bm{\delta u}(t)\Vert} &\leq \frac{e^{t^{1-\alpha}}}{\Vert\bm{\delta u}(0) \Vert e^{\lambda_t t}} \label{eq:ratio_bound_1} \\
    &= \frac{1}{\Vert\bm{\delta u}(0) \Vert} e^{-t(\lambda_t - t^{-\alpha})}.
\end{align}
The bound on the ratio reaches its maximum value
\begin{align}
    r_{\textrm{max}}(\alpha, \lambda_t) = \frac{1}{\Vert\bm{\delta u}(0) \Vert} e^{-t_\textrm{max}(\lambda_t - t_\textrm{max}^{-\alpha})}, \label{eq:sol_sep_ratio_max}
\end{align}
at $t_\textrm{max} = \left[(1-\alpha)/\lambda_t\right]^{1/\alpha}$. The time to reach the maximum value is only dependent on $\alpha$ and $\lambda_t$, not the initial separation $\Vert \bm{\delta u}(0)\Vert$.
At a later time $t_c$, we have $t_c^{-\alpha} = \lambda_t$, after which the ratio $\Vert \bm{u}^A(t)\Vert/\Vert\bm{\delta u}(t)\Vert$ decays exponentially with time. The time $t_c$ also does not depend on the initial separation $\Vert \bm{\delta u} (0)\Vert$. The ratio can now be bounded by an exponential decay. After the time $t>t_c$, the decay rate of the ratio actually increases with time because the term $t^{-\alpha}$ decreases. We can therefore construct a bound by taking an exponential that passes through the maximum $r_{\textrm{max}}(\alpha, \lambda_t)$ at a time $t>t_c$ with a decay rate corresponding to this time. For numerical simplicity, and the fact that a loose bound is satisfactory, we choose $t=2 t_c$ with a decay rate $\lambda_t(1-2^{-\alpha})$ to find the exponential upper bound,
\begin{align}
    \frac{\Vert \bm{u}^A(t)\Vert}{\Vert\bm{\delta u}(t)\Vert} &\leq  \frac{r_{\textrm{max}}(\alpha, \lambda_t)}{e^{-2\lambda_t^{1-1/\alpha}(1-2^{-\alpha})}}  e^{-\lambda_t (1-2^{-\alpha}) t}\\
    &\leq \frac{c}{\Vert \bm{\delta u}(0)\Vert} e^{-\lambda_t (1-2^{-\alpha}) t}, \label{eq:ratio_upper_bound}
\end{align}
where in the second line we have used Eq.~\eqref{eq:sol_sep_ratio_max} and with 
\begin{align}
    c = \exp\{-t_\textrm{max}(\lambda_t - t_\textrm{max}^{-\alpha})+2\lambda_t^{1-1/\alpha}(1-2^{-\alpha})\}.
\end{align}
Fig.~\ref{fig:chaos_upper_bound} demonstrates the bound for specific $\alpha$ and $\lambda_t$. 
\begin{figure}[t]
    \centering
    \includegraphics[scale=0.38]{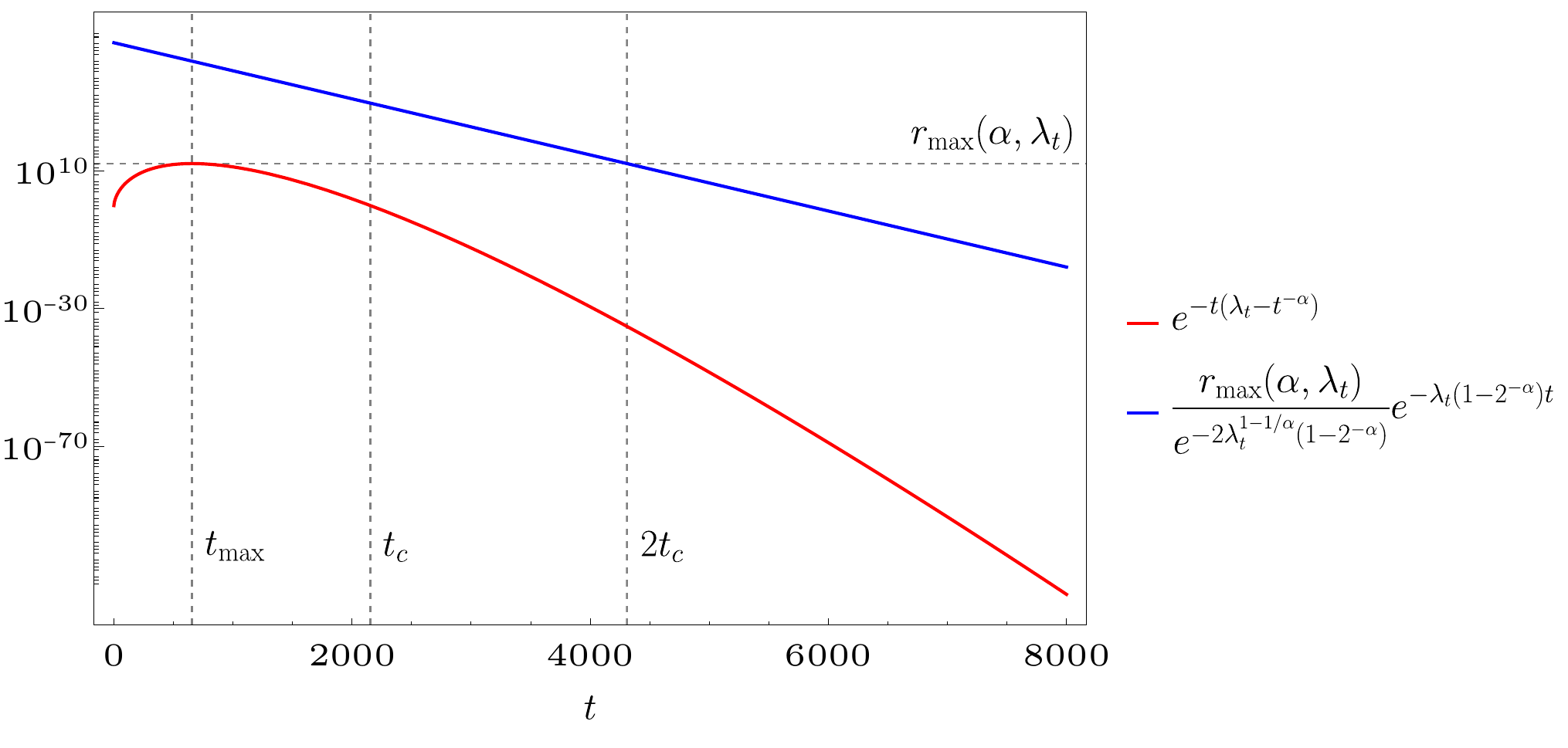}
    \caption{The ratio of Eq.~\eqref{eq:ratio_bound_1} (red) is bounded by the construction of Eq.~\eqref{eq:ratio_upper_bound}
    \label{fig:enter-label} (blue) for an example with $\alpha = 0.3$ and $\lambda_t = 0.1$. The vertical dashed lines show values of $t$ equal to $t_\textrm{max}$, $t_c$, and $2t_c$. The horizontal dashed line is at $r_\textrm{max}(\alpha,\lambda_t)$.}
    \label{fig:chaos_upper_bound}
\end{figure}
We now use the fact that the separation of two trajectories in the chaotic region increases exponentially up to the size of the attractor. Since the chaotic region surrounds the origin, the norm of the separation must approach the norm of the vector. The ratio for the norm is upper bounded by Eq.~\eqref{eq:ratio_upper_bound}, hence the time for the upper bound to reach $1$ gives an upper bound for the time for the ratio $\Vert \bm{u}^{A}(t^*) \Vert/ \Vert \bm{\delta u}(t^*)\Vert$ to reach $1$. Thus, after a time $t^* = O(\log(1/\Vert \bm{\delta u}(0) \Vert))$, we have $\Vert \bm{u}^{A}(t^*) \Vert/ \Vert \bm{\delta u}(t^*) \Vert = 1$. This gives 
\begin{align}
     \Vert \bm{u}^{A}(t^*) \Vert = \Vert \bm{\delta u}(0) \Vert e^{\lambda t^*}. 
\end{align}
Note, we have implicitly assumed that the time $t^* \gg 2 t_c$. $t^*$ increases as the initial separation $\Vert \bm{\delta u}(0)\Vert$ decreases, but $t_c$ does not depend on the initial separation, only the Lyapunov exponent $\lambda_t$ and the sub-exponential bound parameter $\alpha$. Thus, we can always choose a sufficiently small $\Vert \bm{\delta u}(0) \Vert$ such that $t^* \gg 2 t_c$.
Without loss of generality, we assume $\Vert \bm{u}^{A}(t^*) \Vert \geq \Vert \bm{u}^{B}(t^*)\Vert$. After $t^*$, we have 
\begin{align}
    \langle u^A(t^*)|u^B(t^*) \rangle &= \frac{\Vert \bm{u}^B(t^*) \Vert^2}{2 \Vert \bm{u}^A(t^*) \Vert \Vert \bm{u}^B(t^*) \Vert} \\
    &\leq \frac{1}{2}. 
\end{align}
We have a large initial overlap $\langle u^A(0)| u^B(0)\rangle = 1-\varepsilon$, which means, from Eq.~\eqref{eq:initial_separation},
\begin{align}
    \varepsilon = \frac{ \Vert \bm{\delta u}(0)\Vert^2}{2}.
\end{align}
The overlap is reduced to at most $1/2$ in a time $t^* = O(\log(1/\varepsilon))$. We can therefore apply Lemma~\ref{lemma:time_to_reduce_overlap}. A quantum algorithm that gives a quantum state proportional to the solution for a chaotic system with an attractor that surrounds the origin must therefore scale exponentially in time.
\end{proof}

Theorem~\ref{theorem:no_algo_for_choas} does not encompass the result of Theorem~\ref{theorem:R_greater_than_1} because chaos does not exist in two-dimensional systems by Poincaré–Bendixson theorem~\cite{teschl_ordinary_nodate}. In Section~\ref{sec:stable_chaos} we give an example of a dynamical system where, due to Theorem~\ref{theorem:no_algo_for_choas}, we can prove there is no efficient quantum algorithm for producing quantum states proportional to the solution of the system.

\section{\label{sec:stable_chaos}Chaos with a stable fixed point}
Theorem~\ref{theorem:no_algo_for_choas} applies to any quantum algorithm that approximately prepares a normalised quantum state with amplitudes proportional to the solution vector. The QCL algorithm only works for dissipative systems with the real part of the eigenvalues of the linear $F_1$ term being strictly negative, $\textrm{Re}(\lambda_N) < 0$. We construct a system of differential equations that show chaotic solutions for a large region of initial conditions but also satisfy this condition on $F_1$ eigenvalues---a stable fixed point that exhibits chaos. By changing only the initial condition, the behaviour of the system changes from a chaotic solution to a solution that decays smoothly. We prove that no quantum algorithm can efficently output a state proportional to the solution of a differential equation with chaos or turbulence when the solution norms grow sub-exponentially. Numerical simulation shows that the QCL algorithm can efficiently solve the smoothly decaying solution. While Ref.~\cite{liu_efficient_2021} proves that systems with $R<1$ can be efficiently solved, it does not rule out any $R\geq1$ value. In particular, Ref.~\cite{liu_efficient_2021} found an example with $R\approx 44$ for which Carleman linearisation error appears to be bounded for arbitrary time. In the following example, however, we show how the Lyapunov exponent can be used to rule out efficient quantum simulation for a regime with $R\geq 1$. 

We introduce a slightly modified system of differential equations for 
\begin{align}
    \bm{u}(t) = (x(t),y(t),z(t)) \in \mathbb{R}^3
\end{align}
proposed in Ref.~\cite{m_d_dynamical_2021},
\begin{equation}
\begin{gathered}
    \frac{dx}{dt} = y \\
    \frac{dy}{dt} = z \\
    \frac{dz}{dt} = -x -(1-k)y - z -2.3z^2 + xy +k,
\end{gathered}
\end{equation}
with $k\in \left[-1,0\right)$.
This system exhibits chaos with a stable fixed point at $\bm{u}_\mathrm{FP} = (k,0,0)$. The linear term is given by
\begin{equation}
    F_1 = \begin{pmatrix}
        0 & 1 & 0 \\
        0 & 0 & 1 \\
        -1 & k-1 & -1
    \end{pmatrix},
\end{equation}
which has eigenvalues $\{ \lambda_1, \lambda_2, \lambda_2^*\}$ and corresponding eigenvectors $\{ \bm{v}_1, \bm{v}_2, \bm{v}_2^* \}$, dependent on the choice of $k$. The maximum real part of the eigenvalue of $F_1$ is strictly negative for the stated domain of $k$ and has norm $\vert \mathrm{Re}(\lambda_2^*)\vert = \vert \mathrm{Re}(\lambda_2)\vert = \vert k\vert/4$. The eigenvalue $\lambda_1$, with eigenvector $\bm{v}_1$, has only a real part and is indicative of a decaying solution component.

The maximum singular value of $F_2$ gives $\Vert F_2\Vert = 2.51$ and $\Vert F_0\Vert = \vert k\vert$. Therefore $R$ is given by
\begin{align}
    R = \frac{10.04}{\vert k\vert} \Vert \bm{u}(0) \Vert + \frac{4}{\Vert \bm{u}(0) \Vert}.
\end{align}
We consider two explicit initial conditions: (i) the chaotic initial condition, 
\begin{align}
    \label{eq:A_initial_condition}
    \bm{u}^{\textrm{chaos}}(0) = \begin{pmatrix}
        k-1 \\ 0 \\ 0
    \end{pmatrix}
\end{align}
and (ii) the stable initial condition, using a small perturbation by $\xi = 5\times10^{-3}$ in the direction of the decaying eigenvector from the fixed point, 
\begin{align}
    \label{eq:B_initial_condition}
    \bm{u}^{\textrm{decay}}(0) = \begin{pmatrix}
    k + \xi \bm{e}_1 \cdot \bm{v}_1 \\ 
    \xi \bm{e}_2 \cdot \bm{v}_1 \\ 
    \xi \bm{e}_3 \cdot \bm{v}_1
    \end{pmatrix},
\end{align}
where $\bm{e}_1 = (1,0,0)$, $\bm{e}_2 = (0,1,0)$, and $\bm{e}_3 = (0,0,1)$.
For $k=-10^{-4}$, we find $R^{\textrm{chaos}} = 1.00 \times 10^5$ and $R^{\textrm{decay}} = 1.305 \times 10^3$. It is therefore not known a priori for these initial conditions whether the error of the QCL algorithm remains bounded for all times since $R>1$. However, using the results of Section~\ref{sec:positive_lyapunov_no_advantage}, we can prove that no quantum algorithm can efficiently give a quantum state proportional to the solution for all initial conditions within the chaotic attractor for initial condition (i). 
\begin{figure*}
    \centering
    \includegraphics[scale=0.407]{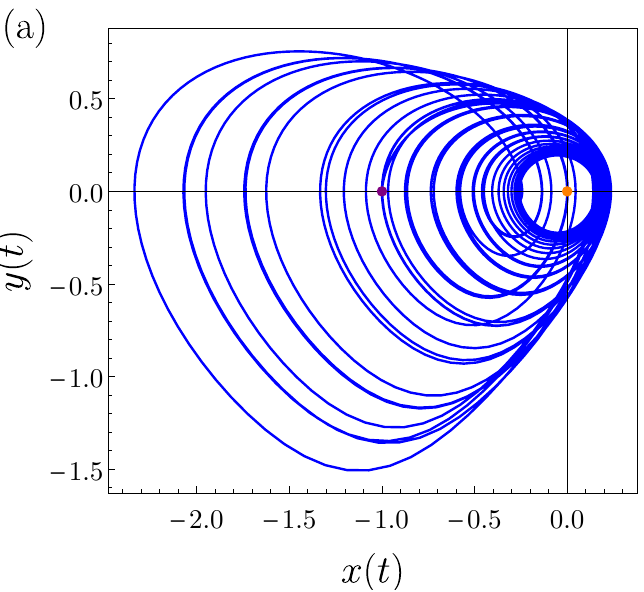}
    \includegraphics[scale=0.407]{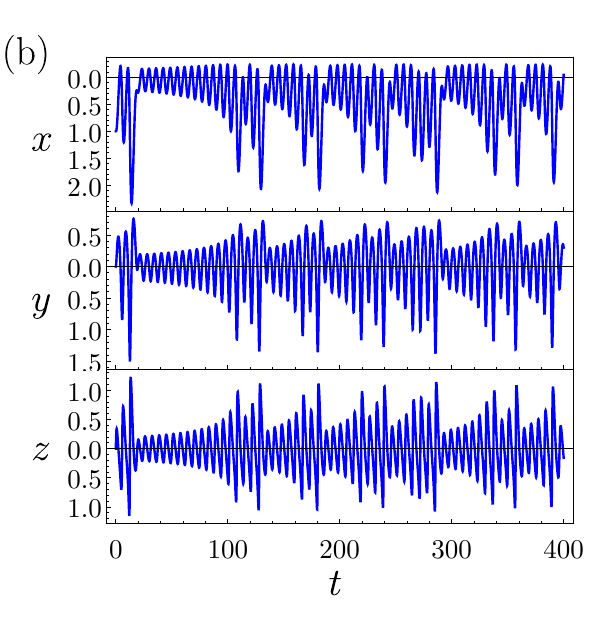}
    \includegraphics[scale=0.407]{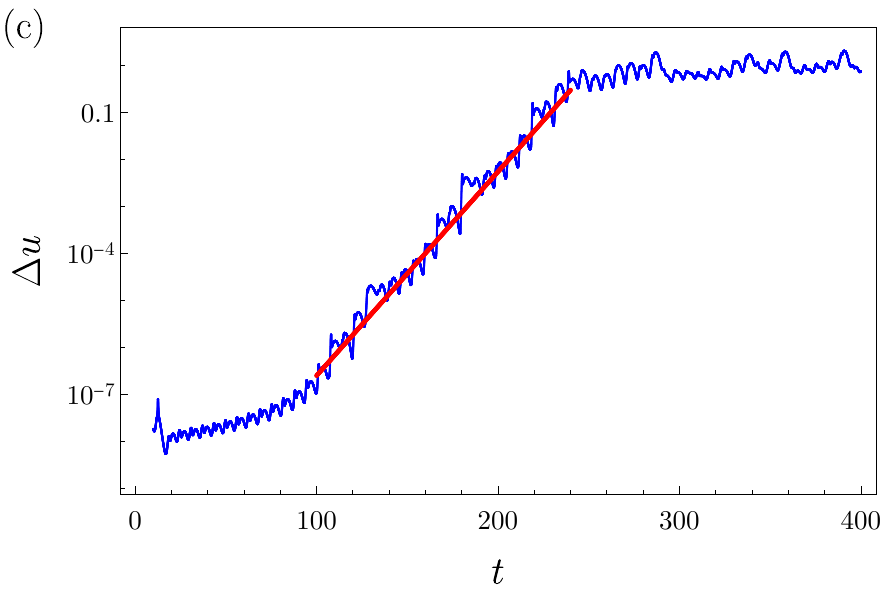}
    \caption{The solution up to time t=400 for initial condition $\bm{u}^{\textrm{chaos}}(0) = (k-1,0,0)$ for $k=-10^{-4}$. (a) shows a parametric plot of the trajectory of the $x$-$y$ plane over time with the initial condition marked as a purple dot and the fixed point as an orange dot. (b) plots the chaotic behaviour of the solution components of $\bm{u}^{\textrm{chaos}}(t)$. (c) plots the average of 100 samples of $\Delta u^{\textrm{chaos}}(t) = \Vert\bm{\Delta u}^{\textrm{chaos}}(t)\Vert$, where $\bm{\Delta u}^{\textrm{chaos}}(t) = \bm{F}_t (\bm{u}^{\textrm{chaos}}(0)) - \bm{F}_t (\bm{u}_\varepsilon$), $\bm{F}_t(\bm{u}(0))$ is the flow function to update from the initial state to the state at time $t$, and $\bm{u}_\varepsilon = \bm{u}(0) + \bm{\varepsilon}$ with $\bm{\varepsilon}$ a random vector with magnitude $\vert \bm{\varepsilon} \vert = 10^{-8}$. The red line is an approximate fit of the function $\sim e^{\lambda t}$ over the period of exponential separation of trajectories from close initial states, giving a Lyapunov exponent of $\lambda = 0.1$.}
    \label{fig:chaotic_solution}
\end{figure*}
Fig.~\ref{fig:chaotic_solution}(a) shows the $x(t)$-$y(t)$ phase plot of the chaotic solution and Fig.~\ref{fig:chaotic_solution}(b) shows how each component of the solution varies up to time $t=400$. Fig.~\ref{fig:chaotic_solution}(c) finds the Lyapunov exponent for the solution to be approximately 0.1. The solution norm does not increase exponentially because the eigenvalues of the system are negative in the range of $k$. Thus, we can apply the result of Theorem~\ref{theorem:no_algo_for_choas}.
\begin{figure*}
    \centering
    \includegraphics[scale=0.302]{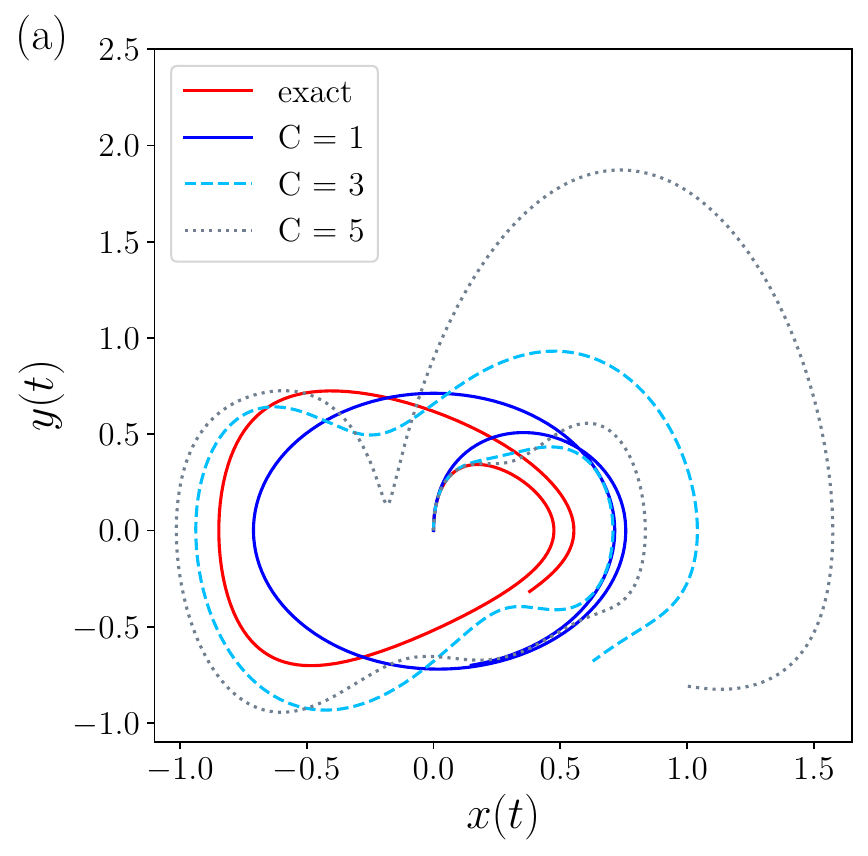}
    \includegraphics[scale=0.302]{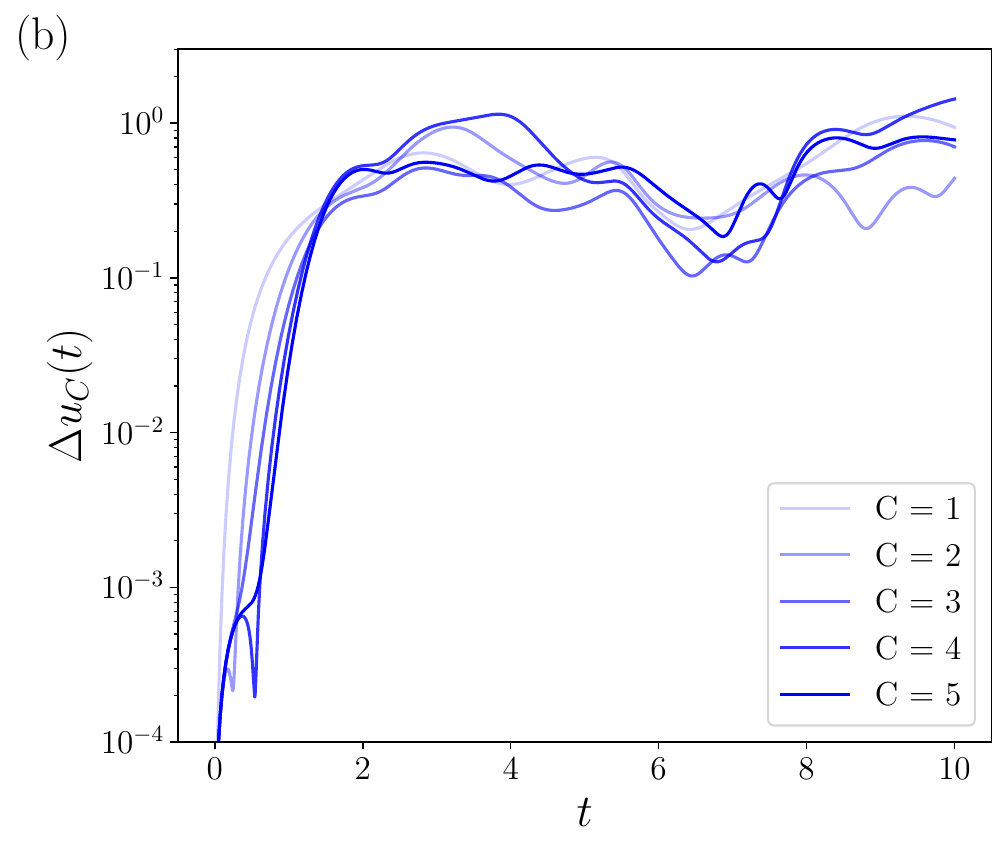}
    \includegraphics[scale=0.302]{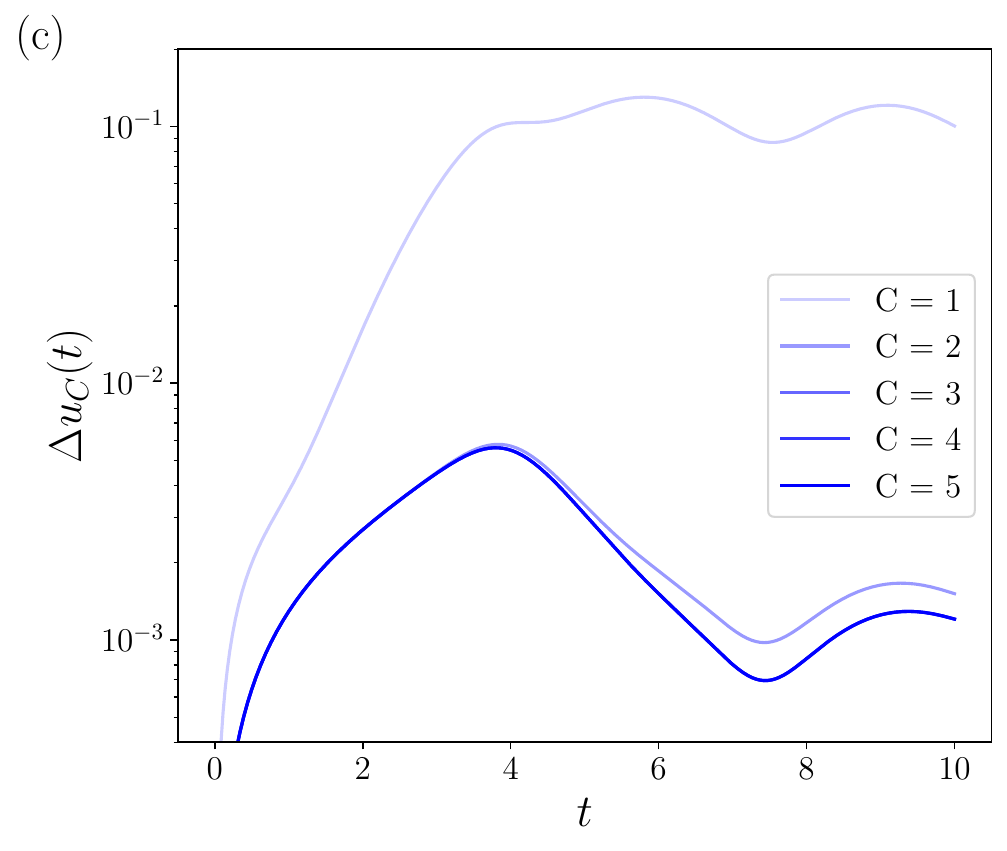}
    \caption{Carleman linearisation up to time t=10 for initial conditions (i) of Eq.~\eqref{eq:A_initial_condition} and (ii) of Eq.~\eqref{eq:B_initial_condition} with $k = -10^{-4}$. (a) shows the phase portrait for the \emph{exact} solution (high precision numerical calculation using Mathematica) compared to Carleman linearisation solutions for truncation orders $C=1,3,5$ for initial condition (i). (b) shows the error from the exact solution for  Carleman truncation order up to $C=5$ for the initial condition (i); the error from the exact solution is computed by $\Delta u_C (t) = \left\Vert \frac{\bm{u}_C(t)}{\Vert \bm{u}_C(t) \Vert} - \frac{\bm{u}(t)}{\Vert \bm{u}(t) \Vert} \right\Vert$, where $\bm{u}(t)$ is the exact solution and $\bm{u}_C(t)$ is the Carleman linearisation solution for truncation order $C$. The error is calculated using the norms of the exact solution and the Carleman linearisation solution as this is the error relevant for the output of the QCL algorithm. (c) shows the error from the exact solution for  Carleman truncation order up to $C=5$ for the initial condition (ii). After a time $t=8$, the error slightly rises and oscillates as the eigenvectors of the periodic solutions begin to contribute on the slow manifold as the solution decays more slowly towards the fixed point.}
    \label{fig:carleman_linearisation}
\end{figure*}

Fig.~\ref{fig:carleman_linearisation} shows the Carleman linearisation of the QCL algorithm. As predicted, the error for the chaotic initial condition (i) does not decrease by increasing Carleman truncation number $C$ and the error very quickly becomes unbounded. On the other hand, we find that for the decaying initial condition of (ii) there is bounded error for all times even for low truncation order, which suggests the QCL algorithm applies initially despite the large value of $R$. We note that after a time of about $t=8$ the error rises slightly and oscillates as the eigenvectors with complex eigenvalues become relevant as the solution decays on the slow manifold.

\section{\label{sec:conc}Discussion}
We have shown that any quantum algorithm that approximates a state proportional to the solution vector of a nonlinear differential equation (with $\textrm{Re}(\lambda_N) < 0$) is not efficient in the worst case for $R \geq 1$. This is achieved by finding a particular coupled differential equation that cannot be solved efficiently on a quantum computer for any $R \geq 1$. The proof shows that if it were possible to find the solution of these nonlinear differential equations then quantum states could be discriminated with fewer resources than required by known lower-bounds on state discrimination task. The intuition for the result stems from the fact that quantum mechanics is linear and quantum operations cannot increase the trace distance between states. Simulating significant nonlinearity efficiently is therefore bounded by these constraints. 

However, determining systems or regimes in systems that are \emph{too nonlinear} is not trivial. Liu et al.~\cite{liu_efficient_2021} introduced the quantity $R$, a Reynolds-like number, that allows the statement that all equations with $R<1$ can be solved efficiently. However, the inverse statement that $R\geq1$ cannot be solved efficiently is not true. In fact, there are examples for systems of even very large $R$, that can be solved with exponential advantage over current classical algorithms and in Section~\ref{sec:stable_chaos} we find a regime with $R = 1.305\times10^3$ that permits Carleman linearisation for arbitrary time. However, we also note that the value of $R$ is dependent on the choice of discretisation. An alternative definition of $R$ has been found for reaction-diffusion equations~\cite{liu_efficient_2023} that removes the effect of discretisation. In this case, the nonlinear terms do not contain derivatives.

A relevant question for solving differential equations on quantum computers is whether there are larger classes of differential equations that can be proven to be not efficiently solvable. A clear first place to look is differential equations with chaos, that is having a positive Lyapunov exponent. Chaos could be considered extremely \emph{nonlinear} since close initial points diverge exponentially fast. In light of this, we prove that, given a natural coordinate system, there cannot be an efficient quantum algorithm that outputs a quantum state proportional to a solution if the solution grows sub-exponentially and exhibits chaos. In Section~\ref{sec:stable_chaos}, we find a system that exhibits chaos for most initial conditions. However, we show that there is also an initial condition that does not exhibit chaos and permits Carleman linearisation with small error, see Fig.~\ref{fig:carleman_linearisation}(c). While all eigenvalues of $F_1$ must be negative to apply the QCL algorithm, it is also possible to have chaos in this regime as demonstrated in Section~\ref{sec:stable_chaos}. In Ref.~\cite{krovi_improved_2023}, the log-norm, defined as the maximum eigenvalue of $\frac{1}{2}(F_1 + F_1^\dagger)$, can also be positive despite all eigenvalues of $F_1$ being negative -- the log-norm can be thought of as providing a bound on the growth of solutions for early times. Chaotic solutions require that at all times solution trajectories separate exponentially (up to the size of the attractor basin), which could indicate a positive log-norm. The system in Section~\ref{sec:stable_chaos} has this feature, where the log-norm is positive but all eigenvalues of $F_1$ are negative. It is possible the log-norm being positive is always a feature of these chaotic systems with stable fixed points.

A further question is whether introducing measurement in the algorithm can simulate chaos or turbulence. Incorporating measurement, and the measurement outcomes, is precisely what allows a general nonlinear transformation in a quantum circuit~\cite{holmes_nonlinear_2023}, which requires exponential scaling of copies of the quantum state for increasing circuit depth. Given the assumptions stated for Theorem~\ref{theorem:no_algo_for_choas}, our proof does indeed rule out the possibility of a polynomial number of copies of the quantum state to simulate chaos or turbulence. Measurement is a state projection that turns amplitudes into statistical uncertainties. Thus, measurement cannot help in distinguishing $\varepsilon$-close states and overcoming the bound for state separation, given in Lemma~\ref{lemma:helstrom_bound}, without the same scaling in the number of copies.

In some cases, it is possible that a nonlinear system can be transformed to a linear system via a classical or analytical mapping. However, chaos and turbulence cannot exist in linear systems. The map itself must be nonlinear, in general, being dependent on the complete current state of the system at all times. All the complexity of the problem would therefore be contained in the map. Moreover, from the perspective of quantum computing, in such an approach one would be solving a linear system using a quantum computer. When we refer to limitations for quantum algorithms to solve turbulent and chaotic systems, we exclude this possibility of reducing the dynamics to linear with classical preprocessing.

We propose that simulation of the classical dynamics of turbulent or chaotic systems is likely not amenable to an exponential quantum advantage over current classical algorithms. This work does not preclude potential for polynomial speedup for quantum algorithms over classical algorithms, for example, turbulent mixing with Monte Carlo simulation has been found to permit a quadratic advantage~\cite{xu_turbulent_2018}. Future work could also investigate nonlinear differential equations with no restrictions on the eigenvalues of $F_1$, however, it is unlikely in general that linearisation of these systems for arbitrary time would be possible since the solutions could be exponentially growing. The case of negative Lyapunov exponents can also be considered. Our results do not imply that all systems with negative maximum Lyapunov exponent would necessarily be amenable to an efficient quantum algorithm. In fact, we consider this unlikely. Future work could also determine further properties of the solutions of PDEs that allow or preclude efficient quantum algorithms.

\section{Acknowledgements}
We thank the anonymous referees for careful reading of our paper, their helpful suggestions, and identifying a technical error in Theorem~\ref{theorem:R_greater_than_1}.
We acknowledge support from the Beyond Moore’s Law project of the Advanced Simulation and Computing Program (ASC) at LANL. DL acknowledges support from the EPSRC Centre for Doctoral Training in Delivering Quantum Technologies, grant ref. EP/S021582/1 and the Quantum Computing Summer School 2022 at Los Alamos National Laboratory (LANL). BN was supported by the US Department of Energy's Scientific Discovery through Advanced Computing (SciDAC) program under projects (a) Non-Hydrostatic Dynamics with Multi-Moment Characteristic Discontinuous Galerkin Methods and (b) Improving Projections of AMOC and its Collapse Through advanced Simulations and the Advanced Machine Learning project of the ASC program at LANL. LA-UR-24-24916.

\end{document}